\documentclass[review,3p,10pt,a4paper,times]{elsarticle}
\usepackage{enumitem}   
\usepackage{epstopdf}
\usepackage{subfigure}
\usepackage[linesnumbered,ruled]{algorithm2e}
\usepackage{multirow}
\usepackage{amssymb}
\usepackage{amssymb,amsmath,amsthm,bm,float}
\usepackage{textcomp}
\usepackage{verbatim}
\usepackage{float}
\usepackage{hyperref}
\usepackage{color}
\newtheorem{remark}{Remark}
\newtheorem{lemma}{Lemma}

\usepackage{mathtools}
\usepackage{caption}
\usepackage{setspace}
\DeclarePairedDelimiter{\ceil}{\lceil}{\rceil}
\SetKwRepeat{Do}{do}{while}
\journal{arXiv}
\begin{document}

	\begin{frontmatter}
		
		
		
		\title{A branch, price and remember algorithm for the U shaped assembly line balancing problem}
		
		\author[majid]{Abdolmajid Yolmeh\corref{cor1}}
		\ead{abdolmajid.yolmeh@rutgers.edu}
		\author[Najmeh]{Najmeh Salehi}
		\ead{tuf62199@temple.edu}
		
		\cortext[cor1]{Corresponding author}
		
		\address[majid]{Industrial and Systems Engineering Department, 
			Rutgers University, 96 Frelinghuysen Rd, Piscataway, NJ 08854}
		
		\address[Najmeh]{Department of Mathematics, Temple University, 1805 N Broad St, Philadelphia, PA 19122}
		
		\begin{abstract}
	In this paper we propose a branch, price and remember algorithm to solve the U shaped assembly line balancing problem. Our proposed algorithm uses a column generation approach to obtain tight lower bounds for this problem. It also stores generated columns in memory to enhance the speed of column generation approach. We also develop a modification of Hoffman algorithm to obtain high quality upper bounds. Our computational results show that our proposed algorithm is able to optimally solve 255 of Scholl's well-known  269 benchmark problems. Previous best known exact algorithm, ULINO, is able to solve 233 of the 269 benchmark problems. We also examined our algorithm on a new data set and the results show that our algorithm is able to solve 96.48 percent of all available benchmark problems.
		\end{abstract}
		
		
	\end{frontmatter}

\section{Introduction}\label{intro} 
When high volume production of a standardized commodity is required, assembly lines are appropriate. Assembly lines consist of a set of $ m $ stations ordered along a conveyer belt or a similar material handling system that moves the work-pieces through them. Starting from the first station, the work-piece enters each station and stays there for a fixed time span called cycle time $ c $. During this time, a set of tasks is performed on the work-piece, then the work-piece is moved to the next station to perform the next set of tasks. This process continues until the final product is assembled in the last station. Traditionally the stations are arranged in a straight line and each work-piece moves along this line, such a system is called a straight assembly line. 

The best known problem in the literature of assembly line balancing is the Simple Assembly Line Balancing Problem (SALBP) \cite{scholl1999balancing,baybars1986survey}. SALBP is defined as follows: 
\begin{itemize}
	\item   A single product is to be manufactured in large quantities in a straight assembly line. The total work required to produce this product, is divided into a set of $ n $ basic tasks $ T=\{1,2,\dots,n\} $. Each task  $ j\in T$ has a positive integral processing time of $ t_j $. 
	\item Due to technological constraints, a task cannot be started before all of its predecessors are finished. The precedence constraints can be represented by a graph, vertices of this graph represent the tasks and precedence relations are represented by directed arcs. An arc $ (i,j) $ means that task $ i $ is a predecessors of task $ j $. We denote the set of immediate (all) successors of task $ j $ by  $ F_j (F_j^* ) $ and the set of  immediate (all) predecessors of task $ j $ by $ P_j (P_j^* ) $. Figure \ref{figre1} shows an example of a precedence graph with $ n=11 $ tasks. In this graph, task processing times are also represented above the corresponding vertices. 
	\item Assignment of tasks to stations must observe precedence constraints. In other words, task $ j $ cannot be assigned to stations $ k $ unless all of its predecessors are assigned to one of the stations $ 1,2,\dots,k $. 
	\item Each task must be assigned to exactly one station. The set of tasks assigned to a station $ k $ is called a work-load of station $ k $ and denoted by $ S_k $. Sum of processing times of tasks assigned to a station $ k $ is called station time and denoted by $ t(S_k ) $. Station times must always observe the cycle time constraint i.e. $ t(S_k )=\sum_{j\in S_k} t_j \le c $.
\end{itemize}
\begin{figure}[H]
	\centering	
	\includegraphics{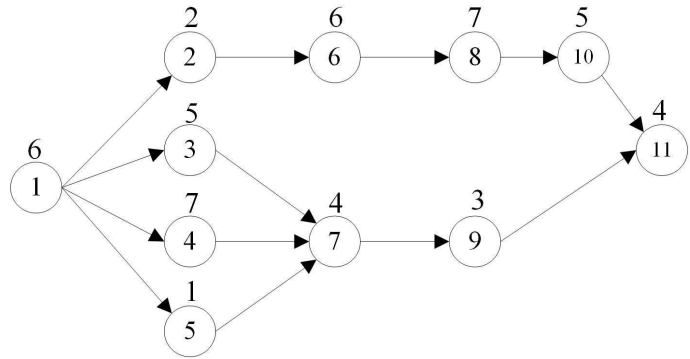}
	\caption{An instance of SALBP-1}
	\label{figre1}
\end{figure}
There are four versions of SALBP: 
\begin{itemize}
	\item SALBP-1: minimizing number of stations $ m $ for a fixed cycle time $ c $.
	\item SALBP-2: minimizing cycle time $ c $ for a given number of stations $ m $. 
	\item SALBP-E: maximizing line efficiency $ E=\frac{(\sum_{j=1}^n t_j )}{mc} $. 
	\item SALBP-F: feasibility problem for given values of $ c $ and $ m $, determining whether or not there is a feasible balance, if so, the problem involves finding a feasible balance. 
\end{itemize}
All of four versions of SALBP are known to be NP-Hard \cite{wee1982assembly}. The assumptions of SALBP are very restricting in comparison with the real world assembly line systems. Therefore, researchers recently have focused on identifying and modeling more realistic situations in assembly lines. The resulting problems are called generalized assembly line balancing problems (GALBP). Several generalizations have been studied for the SALBP. Some examples of GALBPs are assembly lines with resource constraints \cite{augpak2005assembly,corominas2011assembly}, assembly lines with setup times between tasks \cite{andres2008balancing,yolmeh2012efficient}, assembly lines with task deterioration and learning effect \cite{toksari2010assembly,shahanaghi2010scheduling}, assembly lines with parallel workstations \cite{buxey1974assembly}, assembly line balancing and supply chain design \cite{paksoy2012supply,yolmeh2015outer} and assembly lines with multi-manned workstations \cite{dimitriadis2006assembly,fattahi2011mathematical}. Some of the latest surveys of assembly line balancing problems are \cite{erel1998survey,rekiek2002state,scholl2006state}. 

The surveys suggest that, while many applicable problems have been identified and modeled, developing advanced solution methods for these models still lags behind. Even though there have been significant algorithmic developments to solve the SALBP, the same cannot be said about GALBPs. Therefore, the surveys call for additional research to apply modern solution concepts such as advanced enumeration and bounding techniques to solve the generalized models. In this paper we make the first attempt toward developing a branch, price and remember (BP\&R) algorithm to solve the U shaped line balancing problem (UALBP). We use column generation to obtain tight bounds for this problem. We also use memory to avoid visiting redundant sub-problems. To the best of our knowledge, this is the first attempt to apply a branch, price and remember algorithm to solve a U shaped assembly line balancing problem. 

The most closely related work to our study is the branch, bound and remember (BB\&R) algorithm developed by \cite{sewell2012branch} to solve the SALBP. BB\&R combines the idea of branch and bound algorithm with dynamic programming through memorizing already visited nodes in the search tree. Even though our proposed BP\&R also uses memory to avoid redundant nodes, there are some major differences between BP\&R and BB\&R. The first and most important difference is that, unlike BB\&R, BP\&R uses column generation to obtain tight lower bounds for the nodes in the search tree. Column generation algorithm is shown to yield high quality bounds for the SALBP \cite{peeters2006linear}, however its running time was deemed to be too high to be used in a branch and bound algorithm. The main contribution of this paper is devising a scheme to leverage the tight bounds obtained by the column generation algorithm while keeping its run times reasonably low. We do this by relaxing the pricing sub-problem and use of memory. Second difference of BP\&R from BB\&R is in the search strategy. BP\&R uses a best first search strategy, while BB\&R uses cyclic best-first search. Lastly, even though both BB\&R and BP\&R use a modification of Hoffman algorithm to obtain an initial solution, BP\&R uses additional flexibility of U-shaped lines to obtain better solutions. 
\begin{figure}[H]
	\centering
	\includegraphics[scale=1.3]{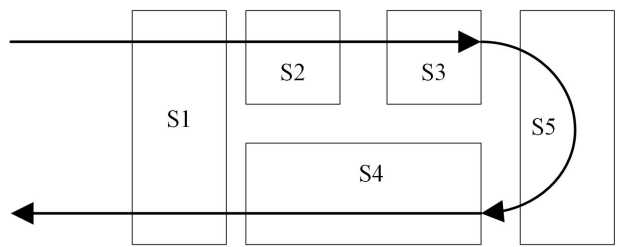}
	\caption{U-shaped assembly line}
	\label{figre2}
\end{figure}
U shaped lines are a generalization of straight lines where the stations can be arranged in a U shape rather than a straight line. This is done by allowing two work-pieces to be processed simultaneously in a single station. By analogy with SALBP, different problem types can be distinguished: UALBP-1, UALBP-2, UALBP-E and UALBP-F. Figure \ref{figre2} shows an example of U-shaped line. In this figure, the first task of one work-piece, which is starting to be produced and the last task of another work-piece is performed in the same station (S1). 
We define the set of forward available tasks as the set of tasks with no unassigned predecessors. In other words, a tasks is called forward available if and only if all of its predecessors have been assigned to a station. Similarly, the set of backward available tasks is the set of tasks with no unassigned successors. The difference between straight and U-shaped lines is that in straight lines, to determine a solution, one should start from the first task and assign tasks into stations while moving forward through the precedence graph i.e. using only forward available tasks. Whereas in a U-shaped line, one can use both forward available and backward available tasks to assign to stations \cite{miltenburg1994u}. Being able to use both forward and backward available tasks to assign to stations results in more possibilities of assigning tasks to stations. Therefore, in comparison with SALBP, UALBP offers more flexibility in assigning tasks to stations, therefore, in many cases, it leads to higher efficiency. 

Consider the example of Figure \ref{figre1} with a cycle time $ c=10 $ and objective of minimizing the number of stations i.e. UALBP-1. Figure \ref{figre3} shows an optimal solution with $ m=5 $ stations for this example problem. Consider station 2 in this figure, we have $ S_2=\{3,10\} $. Task 3 is performed on the work-piece whenever it visits station 2 for the first time, which happens when the work-piece is moving from station 1 to station 2, after its predecessor (task 1) has been executed in station 1. When the work-piece returns to station 2 from station 4, all predecessors of tasks 10 have already been performed (in stations 3 to 5) and task 10 can be executed. After performing task 10, the work-piece moves to station 1 to perform the successor of task 10 i.e. task 11. Then the completed product exits station 1. An optimal solution to the corresponding SALBP-1 requires six stations:
$S_1=\{1,2,5\} $, $ S_2=\{4\} $, $ S_3=\{3,7\} $, $ S_4=\{6\} $, $ S_5=\{8,9\} $ and $ S_6=\{10,11\} $. Therefore, knowing that $ \sum_{j\in T} t_j =50 $, for this example line efficiency of optimal U shaped line is 100\% $ (E=\frac{\sum_{j=1}^n t_j}{mc}=\frac{50}{50}) $ , but line efficiency of optimal straight line is 83.33\% $ (E=\frac{\sum_{j=1}^n t_j}{mc}=\frac{50}{60}) $, which highlights the advantage of U shaped lines over straight lines.
\begin{figure}[H]
	\centering
	\includegraphics{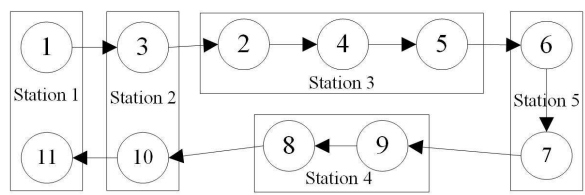}
	\caption{Solution to the example problem with c=10}
	\label{figre3}
\end{figure}
Reviewing the literature of UALBP, \cite{miltenburg1994u} introduced and modeled U-shaped assembly line balancing problem. They developed a dynamic programming approach to solve this problem. \cite{ajenblit1998applying} proposed a genetic algorithm to solve UALBP-1. \cite{urban1998note} developed an integer programming formulation for UALBP-1. \cite{scholl1999ulino} developed a procedure named ULINO (U-Line Optimizer) which is a modification of SALOME \cite{scholl1997salome} to solve all versions of UALBP. ULINO uses a station oriented branching and a depth first search strategy. It also uses several lower bounds and dominance rules. \cite{erel2001balancing} developed a simulated annealing algorithm to solve UALBP-1.
The literature review suggests that, even though many researchers have studied the SALBP, there are few advanced solution methods for the U shaped line balancing problem. The only study that used a branch and bound algorithm to solve this problem is performed by \cite{scholl1999ulino}. They developed an exact algorithm to solve UALBP-1 but also discuss how to use it to solve other versions of the problem. In this paper we develop a branch and price algorithm to solve the UALBP-1. We use a column generation approach to compute tight bounds for this problem. The application of column generation to compute a lower bound is fairly infrequent in the field of assembly line balancing problem. The only study that used column generation to obtain a lower bound for a line balancing problem is \cite{peeters2006linear}. They used Dantzig--Wolfe decomposition to obtain an LP relaxation of SALBP-1 and used column generation algorithm to solve it. To solve the pricing sub-problem, they used a branch and bound algorithm. Their computational results show that, their proposed column generation approach yields very good lower bounds. However, the run times of the column generation was too high to be considered in a branch and bound algorithm as a lower bound. In this paper we tackle the issue of high run times and apply it for the U shaped line balancing problem. Specifically, we relax the sub-problem to obtain a pure knapsack problem, which can be solved in pseudo polynomial time, to obtain a faster column generation. Moreover, we keep a pool of previously generated columns to avoid redundantly generating the same column. Our proposed branch, price and remember also uses memory to avoid visiting redundant nodes in the search tree.

Our proposed BP\&R uses a modification of Hoffman heuristic to obtain a high quality initial solution. Hoffmann algorithm \cite{hoffmann1963assembly} is an efficient heuristic for SALBP-1 that works on a station by station basis. At each iteration, it enumerates all of the possible assignments of tasks to the current station and selects the one with the smallest idle time. This process is repeated until all tasks are assigned. \cite{fleszar2003enumerative} proposed a modification of Hoffman heuristic that works in a bidirectional manner. In other words, their proposed heuristic builds solutions from both sides of the precedence graph and selects the best one.  \cite{sewell2012branch} and \cite{morrison2014application} proposed a modification of Hoffman algorithm that in each iteration, instead of selecting the load with smallest idle time, selects the load with highest value of $ \sum_{j \in U} {(t_j+\alpha w_j+\beta |F_j |-\gamma)} $; where $ F_j (F_j^* ) $ is the set of immediate (all) successors of task $ j $ and $ w_j=t_j+\sum_{k \in F_j^*} t_k $  is the positional weight of task $ j $. $ \alpha,\beta,\gamma $ are coefficients that need to be numerically determined. Our proposed modification of Hoffman heuristic also uses a different criterion to select the load. The difference of our proposed heuristic from the one proposed by  \cite{sewell2012branch} is that our proposed heuristic is designed to be applied on a U shaped line, where forward assignable and backward assignable tasks have different implications on which tasks will be available for the next iteration.

The rest of this paper is organized as follows: in section \ref{sec2} we describe the proposed BP\&R algorithm. Computational results are presented in section \ref{sec3}. In section \ref{sec4}, the main conclusions of the paper and suggestions for future research are given.

\section{The proposed branch, price and remember algorithm}\label{sec2}
BP\&R is a branch, bound and remember (BB\&R) algorithm which uses column generation to obtain lower bounds in each node in the search tree. BB\&R is a branch and bound algorithm that uses memory to avoid revisiting nodes that have already been visited. In BB\&R before branching on a node, it is looked up in the memory to see if it has already been visited. This idea has been used in different fields of combinatorial optimization \cite{jouglet2004branch,morin1976branch}. In the field of assembly line balancing problem, OptPack proposed by \cite{nourie1991finding} to solve SALBP-1 used a variation of BB\&R by storing the sub-problems in a tree structure. \cite{scholl1999balancing} also used a tree structure to store already solved sub-problems. 
\cite{sewell2012branch} used BB\&R to great effect to solve the SALBP-1. For the U shaped line balancing problem \cite{scholl1999ulino} also used memory to store the sub-problems in a tree structure. The next subsections present the details of our proposed BP\&R algorithm. 

\subsection{Branching}
In this section we illustrate the enumeration procedure for our proposed algorithm. The branch and price algorithm explores nodes of an enumeration tree. Each node represents a partial solution. We denote a partial solution as $ \mathcal{P}=(A,U,S_1,S_2,\dots,S_m ) $, where $ A $ and $ U $ are the sets of assigned and unassigned tasks, respectively; $ m $ is the number of stations used by the partial solution and  $ S_1,S_2,\dots,S_m $ are the sets of tasks assigned to stations 1 to $ m $. Based on these definitions we have $ A=\bigcup \limits_{i=1}^m S_i $  and $ U=T\backslash A $. 

Branching on a partial solution (node) means extending it by adding one or many tasks from $ U $ to $ A $ such that the resulting solution is still feasible. There are two main types of branching in the field of assembly line balancing: station oriented branching and task oriented branching. In station oriented branching branches are generated by creating a complete load. However, in task oriented branching, new branches are generated by adding a single task to the current station, if the current station cannot accommodate the task, a new station is created. ULINO, the only branch and bound algorithm available in the literature to solve the U shaped line balancing problem, uses station oriented branching. The proposed BP\&R algorithm also uses station oriented branching for two main reasons: (1) Surveys suggest that station oriented branching tends to yield better results in comparison with task oriented branching. (2) The station oriented branching is a better fit for the BP\&R algorithm. It allows fast computation of column generation lower bound with only minor changes in the restricted master problem.
\subsection{Bounds}
\subsubsection{Upper bound}
Using a heuristic to find a good upper bound can significantly decrease the run time of BP\&R algorithm. We use a modification of the Hoffmann algorithm to obtain a high quality upper bound for the UALBP-1. Hoffmann algorithm \cite{hoffmann1963assembly} is an efficient heuristic for SALBP-1 that generates a solution on a station by station basis. At each iteration, it selects the work-load with smallest idle time. \cite{sewell2012branch} modified the Hoffmann algorithm to search for the load with maximum value of $ \sum_{j \in U} {(t_j+\alpha w_j+\beta |F_j |-\gamma)} $, instead of the load with minimum idle time. In this formula  $ F_j (F_j^* ) $ is the set of immediate (all) successors of task $ j $ and $ w_j=t_j+\sum_{k \in F_j^*} t_k $  is the positional weight of task $ j $. $ \alpha,\beta,\gamma $ are coefficients that need to be numerically determined. 

We propose a Modified Hofmann Heuristic for U shaped lines (MHHU) that works as follows: at each iteration find the load that maximizes $ \sum_{j \in F} {(t_j+\alpha w_j+\beta|F_j |-\gamma)} +\sum_{j \in B} {(t_j+\alpha w_j^{'}+\beta |P_j|-\gamma)} $, where $ F $ is the set of  tasks that are assigned in the forward direction in the U shaped line, $ B $ is the set of  tasks that are assigned in the backward direction in the U shaped line,  $ P_j (P_j^* ) $ is the set of immediate (all) predecessors of task $ j $ and $ w_j^{'}=t_j+\sum_{k \in P_j^*} t_k $  is the backward positional weight of task $ j $. MHHU is called in every node in the BP\&R search tree before branching, to check if the current upper bound can be improved. 
\subsubsection{Lower bounds}
In this section we present the set of lower bounds used in the proposed BP\&R. BP\&R uses three standard lower bounds available in the literature of assembly line balancing ($ LB_1 $ to $ LB_3 $) and a new column generation based bound (CG). $ LB_1 $ to $ LB_3 $ have been previously used in \cite{sewell2012branch}. For each node $ LB_1 $ to $ LB_3 $ are computed first, if they are not able to prune the node, then CG is computed. 
$ LB_1 $ is the total capacity bound which is the simplest available bound in the literature. It follows from the fact that total available time must be greater than or equal to total work content, therefore $LB_1=\ceil{\frac{\sum_{j=1}^{n} t_j}{c}} $ \cite{baybars1986survey}. $ LB_2 $ is obtained by counting the number of tasks $ j $ with $ t_j>\frac{c}{2} $, because these tasks have to be assigned to different stations. 
\[ LB_2=\left| \left\lbrace j \in T: t_j>\frac{c}{2} \right\rbrace \right| + \left\lceil \frac{|\{j \in T: t_j=c/2\}|}{2} \right\rceil \]
$ LB_3 $ is obtained by generalizing $ LB_2 $, this lower bound is obtained by adding task weights that are determined as follows: a weight of 1 is given to all tasks $ j $ with $ t_j>\frac{2c}{3} $, a weight of $ \frac{1}{2} $ is given to all tasks with $ t_j \in (\frac{c}{3},\frac{2c}{3}) $. Tasks $ j $ with $ t_j $ equal to $ \frac{c}{3} $ and $ \frac{2c}{3} $ are given weights of $ \frac{1}{3} $ and $ \frac{2}{3} $, respectively.
\[ 
w_j = \left\{\begin{array}{@{}l@{\quad}l}
1 & \mbox{if $t_j>2c/3$} \\[\jot]
2/3 & \mbox{if $t_j=2c/3$}  \\[\jot]
1/2 & \mbox{if $c/3<t_j<2c/3$} \\[\jot]
1/3 & \mbox{if $t_j=c/3$}
\end{array}\right.
\] 
After computing the weights, $ w_j $, $ LB_3 $ is computed as $ LB_3=\ceil{\sum_{j=1}^{n} w_j} $. 

If $ LB_1 $ to $ LB_3 $ fail to prune the node, column generation is used to obtain a tighter lower bound for the node. The idea behind this approach is to relax the precedence constraints to obtain a bin packing problem and use a column generation approach to obtain a lower bound for the resulting bin packing problem. There are other algorithms to solve the bin packing problem more efficiently. However, we use column generation because the generated columns for one node in the BP\&R search tree can be used to compute the CG lower bound for other nodes. Therefore, after the first run of column generation i.e. at the root node, the subsequent runs will be faster because of the already generated columns. BP\&R maintains a pool of generated columns and uses this pool whenever CG is called at a node.

The objective of column generation is to compute the lower bound for a partial solution $ \mathcal{P}=(A,U,S_1,S_2,\dots,S_m ) $, which corresponds to a node in the BP\&R search tree. After relaxing the precedence constraints, a bin packing problem with the set of unassigned tasks $ U $, as the set of available items (to be fit in a minimum number of bins i.e. stations) is obtained. We define a packing, or a workload or load, as a set of tasks that can be assigned to one station without violating the cycle time constraint. We denote the set of all possible loads by $ L $. Indexing loads by $ l $, and relaxing the precedence constraints we can formulate the UALBP-1 as following:
\begin{flalign}
\label{eq1}
&min \sum_{l\in L} x_l \\ \label{eq2}
&s.t. \quad  \sum_{l \in L} a_{jl} x_l \ge 1 \quad \forall j \in U \\ \label{eq3}
& \quad \quad x_l \in \{0,1\} \quad   \forall l \in L
\end{flalign}

In this formulation, $ x_l $ is a binary variable which is equal to 1 if load $ l $ is selected and 0 otherwise. $a_{jl} $ is a binary parameter that is equal to 1 if task $ j $ belongs to packing $ l $, zero otherwise. The objective function in this formulation is to minimize the number of selected loads. Number of selected loads is equivalent to the number of stations because each load is assigned to one station. Constraint (2) ensures that each unassigned task $ j\in U  $ is in at least one of the selected loads. Constraint (3) is the integrality constraint for variable $ x_l $. We consider the LP relaxation of (\ref{eq1})-(\ref{eq3}):
\begin{flalign}
\label{eq4}
&min \sum_{l \in L} x_l\\\label{eq5}
&s.t. \quad  \sum_{l \in L} {a_{jl} x_l} \ge 1 \quad \forall j \in U	\\ \label{eq6}
& \quad \quad x_l \ge 0 \quad  \forall l \in L 
\end{flalign}

The LP characterized by constraints (\ref{eq4}) to (\ref{eq6}) is called the linear programming master problem (LPM). Note that each column in the LPM corresponds with a work-load. In general the set $ L $, may be exponentially large; however, the number of non-zero variables (the basic variables) in the LPM is equal to the number of constraints, $ |U| $. Therefore, even though the number of possible loads $ L $ is large, only a small number of them is used in the optimal solution. Column generation algorithm uses this idea to start with a subset  $ L^{'} \subseteq L $ of columns and generate columns as needed. The starting subset $ L^{'} $ should be selected such that the following problem is feasible:
\begin{flalign}
&min \sum_{l \in L^{'}} x_l\\ 
&s.t. \quad   \sum_{l \in L^{'}} a_{jl} x_l \ge 1 \quad \forall j \in U	\\
& \quad \quad x_l \ge 0 \quad \forall l \in L^{'}
\end{flalign}

%
%
%
This problem is called Restricted LPM (RLPM). To initialize the column generation algorithm in the root node, we use the loads generated by the MHHU algorithm. Using variable $ \pi_j $ for constraint $j,$ dual of RLPM is as follows:
\begin{flalign}
&max \sum_{j \in U} \pi_j  \\
&s.t. \quad  \sum_{j \in U} {a_{jl}\pi_{j}} \le 1 \quad \forall l \in L^{'}\\
& \quad \quad  \pi_j \ge 0 \quad \forall j \in U
\end{flalign}

%
%
%
Next step is to find a column (load) in $ L\backslash L^{'} $ that could improve the current optimal solution of RLPM. Given the optimal dual solution $ \boldsymbol{\bar{\pi}}=(\bar{\pi}_1,\bar{\pi}_2,\dots,\bar{\pi}_n) $ of (RLPM), the reduced cost of column $ l \in L\backslash L^{'} $ is $ 1-\sum_{j \in U} a_{jl} \bar{\pi}_j$. Based on the concept of duality in linear programming, optimality of RLPM is equivalent with feasibility of the dual. Therefore, loads that violate constraint $ \sum_{j \in U} {a_{jl}\bar{\pi}_{j}} \le 1 $ can improve the current optimal solution. Therefore we should look for a column (load) $ l $ such that: $ 1-\sum_{j \in U} {a_{jl}\bar{\pi}_{j}} <0 $. Note that $ \bar{\pi}_j $ is fixed, and the problem is to find a load $ l $ with $ a_{jl} $ such that: $ 1-\sum_{j \in U} {a_{jl}\bar{\pi}_{j}} <0 $. This problem is called the pricing sub-problem. 
The pricing sub-problem involves finding a set of tasks, a work-load, with a negative reduced cost. Because this work-load does not have to observe the precedence constraints, the problem of finding a work-load with a negative reduced cost is a knapsack problem. To solve this sub-problem, consider the set of unassigned tasks $ U $ and solve the knapsack problem with $ U $ as the set of available items and $ \bar{\pi}_j $ as their values to maximize the overall value $ \sum_{j \in U}{a_{jl}\bar{\pi}_j} $ subject to cycle time constraint. We use the well-known dynamic programing method to solve the generated knapsack problems, which runs in pseudo-polynomial time i.e. $ O(|U|c) $. We also use a dual bound to obtain a lower bound on the number of stations in each iteration

\begin{lemma}\label{lemma1} Let $ v(\text{RLPM}) $ and $ v(\text{LPM}) $ denote the optimum objective function value of the current $ \text{RLPM} $ and $ \text{LPM} $, respectively. Also let $ v(\text{PSP}) $ be the minimum reduced cost obtained by solving the pricing sub-problem to optimality. 
	\begin{enumerate}[label=(\roman*)]
		\item \label{parti} We have: $ v(LPM) \ge \frac{v(RLPM)}{1-v(\text{PSP})}$.
		\item \label{partii} If $ \ceil{v(RLPM)}=\ceil{\frac{v(RLPM)}{1-v(\text{PSP})}} $ we can terminate the column generation algorithm with $ \ceil{v(RLPM)} $ as the obtained lower bound on the number of stations.
	\end{enumerate}

\end{lemma}
\begin{proof}
	General form of part \ref{parti} result can be found in \cite{lubbecke2011column}. During each iteration of the column generation algortihm one cannot improve $ v(\text{\textit{RLPM}}) $ by more than $ \sum_{l\in L} x_l $ times the smallest reduced cost $ v(\text{\textit{PSP}}) $. Moreover, we know that $ \sum_{l\in L} x_l =v(\text{LPM}) $ for an optimal solution of the master problem. Therefore we have:
	$  v(\text{\textit{LPM}}) \ge  v(\text{\textit{RLPM}})+(\sum_{l\in L} x_l)v(\text{\textit{PSP}})=v(\text{\textit{RLPM}})+v(\text{\textit{LPM}})v(\text{\textit{PSP}})$. The part \ref{parti} follows from this inequality.\\
	To prove part \ref{partii} we have: $  \frac{v(RLPM)}{1-v(\text{PSP})} \le v(LPM) \le v(RLPM)$. Because the number of stations should be a whole number, we can use $ \ceil{v(LPM)} $ as a lower bound on the number of stations. Moreover, from $ \ceil{v(RLPM)}=\ceil{\frac{v(RLPM)}{1-v(\text{PSP})}} $, we can conclude that $ \ceil{v(LPM)} = \ceil{v(RLPM)}=\ceil{\frac{v(RLPM)}{1-v(\text{PSP})}}$
\end{proof}

\begin{remark}
	Note that in order for the bound in lemma \ref{lemma1} to be valid, the pricing subproblem should be solved to optimality. In general, the dual bound is not monotone over the iterations, this is called the yo-yo effect \cite{lubbecke2011column}.
\end{remark}

Algorithm \ref{alg1} shows the pseudo-code for the overall column generation algorithm. 
\begin{algorithm}	
	Initialize columns $ \mathcal{L}^\prime $.\\
	Solve RLPM. Let $ \overline{\bm{\pi}}= (\overline{\pi}_1, \overline{\pi}_2, \dots, \overline{\pi}_n) $ be the obtained optimal dual values.\\
	Solve the pricing sub-problem using $ \overline{\bm{\pi}}$ as item values and let $ \textbf{a}^{\textbf{*}}_\textbf{l}=(a_{1l}^*, a_{2l}^*, \dots, a_{nl}^*) $ be the obtained optimal solution.
	
	\eIf{$ 1-\sum_{j \in U} a_{jl}^* \overline{\pi}_j<0  \text{ and } \ceil{\sum_{j \in U}\overline{\pi}_j} \ne \ceil{\frac{\sum_{j \in U}\overline{\pi}_j}{\sum_{j \in U} a_{jl}^* \overline{\pi}_j}}$}{
		Add the new column $ \textbf{a}^{\textbf{*}}_\textbf{l} $ to RLPM.\\
		Go To Line 2.	
	}{
		Return $\sum_{j \in U}\overline{\pi}_j$ as the obtained lower bound.\\
		Terminate the procedure.
	}
	\caption{Pseudo-code for the overall column generation algorithm}
	\label{alg1}
\end{algorithm}
As seen in this pseudo-code, the column generation algorithm starts with a set of initial columns. This set is obtained by the MHHU algorithm in the root node; in other nodes the pool of already generated columns is used as initial set of columns. The RLPM is solved using this set of initial columns and the vector of dual values is obtained. Dual values are then used in the pricing sub-problem to generate a new column. If a new column with a negative reduced cost is obtained and the termination criterion from lemma \ref{lemma1} is not satisfied, it is added to the RLPM and the process is repeated; otherwise the procedure terminates.
\subsection{Use of memory, dominance rules and search strategy }
As its name suggests, BP\&R memorizes already considered nodes. Before considering a node $ \mathcal{N}=(A,U,S_1,S_2,\dots,S_m ) $, BP\&R checks the memory to see if there is a node with the same set of assigned tasks, $ A $, if there is such a node $ \mathcal{M}=(A,U,S_1^{'},S_2^{'},\dots,S_{m^{'}}^{'} ) $ and $ m^{'} \le m $, then $ \mathcal{N} $ is dominated by $ \mathcal{M} $ and $ \mathcal{N} $ can be pruned. OptPack \cite{nourie1991finding}, SALOME \cite{scholl1997salome} and ULINO \cite{scholl1999ulino} applied this dominance rule using a compact tree structure. \cite{sewell2012branch} used a hash table to store the sub-problems. BP\&R Also uses a hash table to store already visited nodes.

Other than the memory based dominance rule, BP\&R uses the maximally loaded dominance rule. This rule states that if a node contains a station that is not maximally loaded it is dominated by a node that contains only maximally loaded stations. 

The third dominance rule used in BP\&R is the Jackson rule which was used in SALOME \cite{scholl1997salome} and BB\&R \cite{sewell2012branch}. ULINO \cite{scholl1999ulino} uses a modification of this rule for a U shaped line, which we will directly apply in our proposed BP\&R.

BP\&R uses a best first search strategy where the nodes with higher number of, fixed, stations have the higher priority. If two nodes have the same number of stations, the one with lower CG lower bound is selected to be branched on. A priority queue is used to handle the generated nodes.
Algorithm \ref{alg2} shows the pseudo-code for the overal BP\&R algorithm. The algorithm starts with setting the root node and computing the upper and lower bounds. If these bounds can prove optimality, i.e. $ LB=UB $, then we are done and $ UB $ can be reported as the optimum number of stations. Otherwise if $ LB<UB $, we initialize the branch and bound tree by initializing a Hash table $ H $ and a priority queue $ PQ $ in lines 5 and 6, respectively. In line 7 we push the root node into the priority queue. The while loop in line 8 starts the best first search process. In lines 9 and 10, we remove a node $ \mathcal{M} $ from the top of the priority queue. Line 11 starts the branching process on node $ \mathcal{M} $ by generating all of the possible children of this node. In line 12, for each generated child we check if it is dominated or if it can be pruned by a lower bound argument. If the child node is dominated or can be pruned based on its computed lower bound, it will be discarded. Otherwise, we compute an upper bound based on node $ \mathcal{C} $ using MHHU approach and try to find an improved upper bound. Lines 19 and 20 add the new child node $ \mathcal{C} $ to the priority queue and the Hash table. Note that Algorithm \ref{alg2} does not reflect the subtlety of order of computing $ LB_1 $ to $ LB_3 $ and column generation bound. As mentioned in the lower bounds section, because $ LB_1 $  to $ LB_3 $ are faster to compute, we compute them first in the hope of pruning the node. If they fail to prune the node, we compute the column generation bound. 

\begin{algorithm}
	Initialize the root node $ \mathcal{N};$\\	
	$ UB :=$ Compute\_upper\_bound $(\mathcal{N});$ \\
	$ LB :=$ Compute\_lower\_bound $(\mathcal{N});$ \\
	\If{$ LB<UB $}{
		Initialize Hash Table $ H ;$\\
		Initialize Priority Queue $ PQ;$\\
		$ PQ.push(\mathcal{N});$\\
		  \While{$ PQ $ is not empty}{
		  	$ \mathcal{M}:=PQ.top(); $\\
		  	$ PQ.pop(); $\\
			\ForEach{child $ \mathcal{C} $ of $ \mathcal{M} $}
			{
				
				\eIf{C is dominated OR Compute\_lower\_bound $(\mathcal{C})$ $ \ge UB$ }{
					Discard $ \mathcal{C};$\\
					Continue\;}{
						\If{Compute\_upper\_bound $(\mathcal{C})$ $ < UB$ }{
						Update $ UB;$}
					$ PQ.push(\mathcal{C});$\\
					$ H \leftarrow \mathcal{C};$
				}

			}			
		}
	}
Output $ UB $ as the obtained optimum number of stations;
	\caption{Pseudo-code for the overall BP\&R}
	\label{alg2}
\end{algorithm}

\section{Computational results} \label{sec3}
In this section we evaluate the BP\&R and CG algorithms in terms of quality of obtained solutions and run times.  The algorithms are coded in C++ and CPLEX 12.6 solver has been used to solve the LPs. The computational experiments are performed on a computer with
Intel Core i7 2.00 GH processor and 8 GB of RAM. We use 269 benchmark problems of SCHOLL data set as well as a new data set generated by \cite{otto2013systematic} to perform the experiments. Both of these data sets are available at the assembly line balancing research homepage (http://alb.mansci.de/). Each instance in the data sets was limited to a total processing time of 500 seconds. 
\subsection{Experiments on SCHOLL data set}

In this subsection we run experiments on the 269 benchmark problems of SCHOLL. To compare the lower bounds $ LB_1 $ to $ LB_3 $ and CG with each other, we run these algorithms on all of the 269 benchmark problems. Table \ref{table0} presents the results of this experiment. In this table,  the relative deviation is computed as $ \frac{UB-LB}{UB} \times 100 \%$, where LB is the obtained lower bound and UB is the best known upper bound. Avg. Rel. Dev. and Max. Rel. Dev. rows are the reported average and maximum relative deviations, respectively. Avg. Abs. Dev. and Max. Abs. Dev. rows are the reported average and maximum absolute deviations, respectively. \#Opt Found row reports the number of cases where the obtained lower bound matched the best known upper bound.

As seen in table \ref{table0}, CG outperforms the standard lower bounds $ LB_1 $ to $ LB_3 $ in all of the performance criteria except for the average CPU time. Moreover, the lower bounds obtained by CG match the best known solution in 238 cases and for the remaining cases, it deviates from the best known solution by only 1 station. This shows that CG obtains tight lower bounds for the UALBP-1, however its average CPU time is relatively high. This may raise questions over its applicability in a branch and bound procedure. To address this question, we consider a modification of BP\&R which does not use CG, we refer to this algorithm as BP\&R without CG (BP\&RWOCG).

{\setstretch{1.0}
\begin{table}[H]
	\centering
	\caption{Comparison of lower bounds}
	{
		\begin{tabular}{ccccc}
			\hline
			& $ LB_1 $ & $ LB_2 $ & $ LB_3 $ & CG \\\hline
			Avg. Rel. Dev. (\%)	&3.39  & 63.13& 46.69 & \textbf{0.85}\\
			Max. Rel. Dev. (\%)	& 37.29 &  100& 100 & \textbf{20} \\
			Avg. Abs. Dev. 	&1.23  & 11.39 & 9.25 & \textbf{0.11}\\
			Max. Abs. Dev. 	& 22 & 37 & 36 & \textbf{1}\\
			\#Opt Found 	& 194 & 13 & 14 & \textbf{239}\\
			Avg. CPU (s)	& \textbf{0.00}  &\textbf{ 0.00} & \textbf{0.00} & 4.33\\ \hline
		\end{tabular}
	}
	\label{table0}
\end{table}	
}

In the next experiment we investigate the performance of BP\&R and BP\&RWOCG. We run these algorithms on  all of the 269 benchmark problems and report the results in table \ref{table1}. In this table  ``Optimal found" counts the number of cases where the optimal solution is obtained. ``Optimal verified" counts the number of cases where the obtained solution is proven to be optimal. ``Avg. dev." and ``Max. dev." represent the average and maximum relative deviation from lower bound. The relative deviation is computed as $ \frac{UB-LB}{LB} \times 100 \%$, where UB and LB are obtained upper and lower bounds, respectively. ``Avg. CPU verified (s)" shows the average run times in seconds on the problems for which optimality was verified. ``Avg. CPU all (s)" shows the average run time, in seconds, for all instances. For 6 cases, BP\&RWOCG runs out of memory, these cases are not considered when computing the average CPU times for BP\&RWOCG. The results of ULINO algorithm are taken from (Scholl and Klein 1999b). They coded their algorithm using Borlands Pascal 7.0. and performed tests 
on an IBM-compatible personal computer with a 80486 DX2-66 central
processing unit and the operating system MS DOS 6.2. They also applied a time limit of 500 seconds per instance (excluding input and output operations). 

\begin{table}[H]
	\centering
	\caption{Comparison of BP\&R with BP\&RWOCG and ULINO}
		\begin{tabular}{cccc}
		\hline
		& ULINO & BP\&R & BP\&RWOCG \\\hline
		Optimal found	& 233 & \textbf{255} & 250\\
		Optimal verified	& 233 & \textbf{255} & 233\\
		Avg. Rel. Dev. (\%)	&0.59  & \textbf{0.26}& 0.60 \\
		Max. Rel. Dev. (\%)	& 10 &  \textbf{7.14}& 11.76 \\
		Avg. Abs. Dev. 	&0.13  & \textbf{0.05} & 0.16 \\
		Max. Abs. Dev. 	& \textbf{1} &  \textbf{1} & 4\\
		Avg. CPU verified (s)	& - & 9.11 & 2.34 \\
		Avg. CPU all (s)	& 82.09  & \textbf{34.66} & 59.11 \\ \hline
\end{tabular}
	\label{table1}
\end{table}

As seen in table \ref{table1}, out of 269 benchmark problems, BP\&R is able to solve and verify optimality of 255 problems. Out of these 255 problems, 216 where closed at the root node. Meaning that, for these problems, CG and MHHU where able to prove optimality without branching.

MHHU was able to find the optimal solution in 232 instances out of 269 benchmark problems. In 36 instances, MHHU deviates from optimal solution by only 1 station. For one instance, MHHU deviates from optimal solution by 2 stations. 

Comparing BP\&R with ULINO, BP\&R outperforms ULINO in terms of number of optimal solutions found and verified, BP\&R also does better than ULINO in terms of average and maximum deviations from lower bound. ULINO never does better than BP\&R, in terms of upper and lower bounds, for any benchmark problem. However, the run times cannot be compared because of different computer used to run ULINO. It is noteworthy to mention that BP\&R was able to verify the optimal solutions in less than 10 seconds on average. 

Comparing BP\&R with BP\&RWOCG, BP\&R outperforms BP\&RWOCG in all performance criteria except for Avg. CPU verified (s), where BP\&RWOCG does better. This highlights the importance of using CG as a lower bound in  BP\&R. Because $ LB_1 $ to $ LB_3 $ are faster than CG, BP\&R uses them first, if they fail to prune a node, CG is called to obtain tighter bounds. For 269 benchmark problems, a total of 4,878,445 nodes where searched. $ LB_1 $ to $ LB_3 $ are called for all of these nodes, out of which 4,335,666 where pruned. Out of the remaining 542,779 nodes CG was able to prune 395,649. On average, CG improved the number of pruned nodes by 9.12 percent.

Table \ref{table2} shows the detailed results for the challenging problems where BP\&R was able to find and verify the optimal solution. Challenging problems are defined to be the ones ULINO was not able to verify optimality. For some of these problems, the optimal solution is verified from SALBP-1 results in the literature. These results are labeled ``Known optimal (or range)" in the table. The instances for which optimality is verified are highlighted in boldface. The open problems solved by BP\&R are determined by a star ``*" besides the results. As seen in table \ref{table2}, BP\&R is able to solve 13 open problems. Moreover, there are 9 problems for which ULINO was not able to verify optimality but the lower bound obtained from solving SALBP-1 verifies optimality of these problems. BP\&R is able to find and verify the optimal solution for these 9 problems without using the SALBP-1 optimal solution as a lower bound. 

Table \ref{table3} shows the results obtained for the remaining 14 problems that BP\&R was not able to solve. As seen in this table, 13 of the unsolved problems are ``Arcus 2" with different cycle times, and 1 problem is ``Scholl" with cycle time of 1422. As seen in this table, for all of the 14 problems, both ULINO and BP\&R have the same performance in terms of upper and lower bounds.
\begin{table}[H]
	\centering
	\caption{Solved challenging problems}
	{%
		\begin{tabular}{ccccccccccccc}
			\hline
			\multicolumn{3}{c}{Problem} &  & \multicolumn{2}{c}{\begin{tabular}[c]{@{}c@{}}Known optimal\\ (or range)\end{tabular}} &  & \multicolumn{2}{c}{ULINO} &  & \multicolumn{3}{c}{BP\&R}                      \\ \cline{1-3} \cline{5-6} \cline{8-9} \cline{11-13} 
			Name        & n     & c     &  & LB                                         & UB                                         &  & LB          & UB          &  & LB           & UB           & CPU              \\ \cline{1-3} \cline{5-6} \cline{8-9} \cline{11-13} 
			Arcus 1     & 83    & 3786  &  & \textbf{21}                                & \textbf{21}                                &  & 21          & 22          &  & \textbf{21}  & \textbf{21}  & \textbf{2.04}    \\
			Arcus 2     & 111    & 11570  &  & \textbf{13}                                & \textbf{13}                                &  & 13          & 14          &  & \textbf{13}  & \textbf{13}  & \textbf{65.64}    \\
			Barthol2    & 148   & 85    &  & \textbf{50}                                & \textbf{50}                                &  & 50          & 51          &  & \textbf{50}  & \textbf{50}  & \textbf{1.45}    \\
			Barthol2    & 148   & 89    &  & \textbf{48}                                & \textbf{48}                                &  & 48          & 49          &  & \textbf{48}  & \textbf{48}  & \textbf{0.83}    \\
			Barthol2    & 148   & 93    &  & \textbf{46}                                & \textbf{46}                                &  & 46          & 47          &  & \textbf{46}  & \textbf{46}  & \textbf{0.80}    \\
			Barthol2    & 148   & 97    &  & \textbf{44}                                & \textbf{44}                                &  & 44          & 45          &  & \textbf{44}  & \textbf{44}  & \textbf{0.84}    \\
			Kilbridge   & 45    & 56    &  & \textbf{10}                                & \textbf{10}                                &  & 10          & 11          &  & \textbf{10}  & \textbf{10}  & \textbf{0.11}    \\
			Mukherje    & 94    & 176   &  & 24                                         & 25                                         &  & 24          & 25          &  & \textbf{24*} & \textbf{24*} & \textbf{0.52*}   \\
			Scholl      & 297   & 1394  &  & \textbf{50}                                & \textbf{50}                                &  & 50          & 51          &  & \textbf{50}  & \textbf{50}  & \textbf{20.78}   \\
			Scholl      & 297   & 1515  &  & \textbf{46}                                & \textbf{46}                                &  & 46          & 47          &  & \textbf{46}  & \textbf{46}  & \textbf{22.96}   \\
			Tonge       & 70    & 160   &  & 22                                         & 23                                         &  & 22          & 23          &  & \textbf{22*} & \textbf{22*} & \textbf{7.49*}   \\
			Tonge       & 70    & 176   &  & 20                                         & 21                                         &  & 20          & 21          &  & \textbf{20*} & \textbf{20*} & \textbf{7.19*}   \\
			Warnecke    & 58    & 54    &  & 30                                         & 31                                         &  & 30          & 31          &  & \textbf{31*} & \textbf{31*} & \textbf{0.89*}   \\
			Warnecke    & 58    & 62    &  & 26                                         & 27                                         &  & 26          & 27          &  & \textbf{26*} & \textbf{26*} & \textbf{3.62*}   \\
			Warnecke    & 58    & 65    &  & 24                                         & 25                                         &  & 24          & 25          &  & \textbf{25*} & \textbf{25*} & \textbf{0.49*}   \\
			Warnecke    & 58    & 68    &  & 23                                         & 24                                         &  & 23          & 24          &  & \textbf{23*} & \textbf{23*} & \textbf{8.60*}   \\
			Warnecke    & 58    & 71    &  & 22                                         & 23                                         &  & 22          & 23          &  & \textbf{22*} & \textbf{22*} & \textbf{155.44*} \\
			Warnecke    & 58    & 74    &  & 21                                         & 22                                         &  & 21          & 22          &  & \textbf{22*} & \textbf{22*} & \textbf{3.95*}   \\
			Warnecke    & 58    & 82    &  & 19                                         & 20                                         &  & 19          & 20          &  & \textbf{19*} & \textbf{19*} & \textbf{7.42*}   \\
			Wee-mag     & 75    & 47    &  & 32                                         & 33                                         &  & 32          & 33          &  & \textbf{32*} & \textbf{32*} & \textbf{5.05*}   \\
			Wee-mag     & 75    & 49    &  & 31                                         & 32                                         &  & 31          & 32          &  & \textbf{32*} & \textbf{32*} & \textbf{0.63*}   \\
			Wee-mag     & 75    & 50    &  & 31                                         & 32                                         &  & 31          & 32          &  & \textbf{32*} & \textbf{32*} & \textbf{0.70*}   \\ \hline
		\end{tabular}%
	}
	\label{table2}
\end{table}

\begin{table}[H]
	\centering
	\caption{Remaining unsolved problems}
	{
		\begin{tabular}{cccccccccccc}
		\hline
		\multicolumn{3}{c}{\multirow{2}{*}{Problem}} & \multirow{2}{*}{} & \multicolumn{2}{c}{\multirow{2}{*}{\begin{tabular}[c]{@{}c@{}}Known range\end{tabular}}} & \multirow{2}{*}{} & \multicolumn{2}{c}{\multirow{2}{*}{ULINO}} & \multirow{2}{*}{} & \multicolumn{2}{c}{\multirow{2}{*}{BP\&R}} \\
		\multicolumn{3}{c}{}                         &                   & \multicolumn{2}{c}{}                                                                                    &                   & \multicolumn{2}{c}{}                       &                   & \multicolumn{2}{c}{}                       \\ \cline{1-3} \cline{5-6} \cline{8-9} \cline{11-12} 
		Name            & n           & c            &                   & LB                                                 & UB                                                 &                   & LB                   & UB                  &                   & LB                   & UB                  \\ \cline{1-3} \cline{5-6} \cline{8-9} \cline{11-12} 
		Arcus 2         & 111         & 5785         &                   & 26                                                 & 27                                                 &                   & 26                   & 27                  &                   & 26                   & 27                  \\
		Arcus 2         & 111         & 6016         &                   & 25                                                 & 26                                                 &                   & 25                   & 26                  &                   & 25                   & 26                  \\
		Arcus 2         & 111         & 6267         &                   & 24                                                 & 25                                                 &                   & 24                   & 25                  &                   & 24                   & 25                  \\
		Arcus 2         & 111         & 6540         &                   & 23                                                 & 24                                                 &                   & 23                   & 24                  &                   & 23                   & 24                  \\
		Arcus 2         & 111         & 6837         &                   & 22                                                 & 23                                                 &                   & 22                   & 23                  &                   & 22                   & 23                  \\
		Arcus 2         & 111         & 7162         &                   & 21                                                 & 22                                                 &                   & 21                   & 22                  &                   & 21                   & 22                  \\
		Arcus 2         & 111         & 7520         &                   & 20                                                 & 21                                                 &                   & 20                   & 21                  &                   & 20                   & 21                  \\
		Arcus 2         & 111         & 7916         &                   & 19                                                 & 20                                                 &                   & 19                   & 20                  &                   & 19                   & 20                  \\
		Arcus 2         & 111         & 8356         &                   & 18                                                 & 19                                                 &                   & 18                   & 19                  &                   & 18                   & 19                  \\
		Arcus 2         & 111         & 8847         &                   & 17                                                 & 18                                                 &                   & 17                   & 18                  &                   & 17                   & 18                  \\
		Arcus 2         & 111         & 9400         &                   & 16                                                 & 17                                                 &                   & 16                   & 17                  &                   & 16                   & 17                  \\
		Arcus 2         & 111         & 10027        &                   & 15                                                 & 16                                                 &                   & 15                   & 16                  &                   & 15                   & 16                  \\
		Arcus 2         & 111         & 10743        &                   & 14                                                 & 15                                                 &                   & 14                   & 15                  &                   & 14                   & 15                  \\

		Scholl          & 297         & 1422         &                   & 49                                                 & 50                                                 &                   & 49                   & 50                  &                   & 49                   & 50                  \\ \hline
	\end{tabular}
}
	\label{table3}
\end{table}

\subsection{Experiments on the new data set}
Otto et al \cite{otto2013systematic} published a new problem database of generated instances for the SALBP. In this subsection we use this data set to evaluate the performance of BP\&R algorithm. Eventhough this data set is originally devised for the SALBP we can use BP\&R to solve the U shaped version of the problems inside this data set.

Table \ref{table_5_otto} shows the results of BP\&R on the new data set. As seen in this table, all of the small instances, 99.73 percent of medium instances, 91.05 percent of large instances and 65.9 percent of very large instances are solved by BP\&R. The fourth column, ``\# MHHU found Optimum", shows the number of instances where MHHU found the optimal solution. The fifth column, ``\# Closed in root node", shows the number of instances that where closed in the root node. In other words, for these instances MHHU finds the optimal solution and the lower bounds prove its optimality without need for further branching.

\begin{table}[H]
	\centering
	\caption{Number of solved problem instances}
		{
	\begin{tabular}{ccccc}
		\hline
		Size       & \# BP\&R Solved & \% BP\&R Solved & \# MHHU found Optimum & \# Closed in root node \\ \hline
		Small      & 525             & 100.00          & 486                & 481                 \\
		Medium     & 5236            & 99.73           & 4118               & 4062                \\
		Large      & 478             & 91.05           & 359                & 359                 \\
		Very Large & 346             & 65.90           & 346                & 346                 \\ \hline
	\end{tabular}}
\label{table_5_otto}
\end{table}

Table \ref{Table_6} shows the run times statistics of BP\&R. As seen in this table, all of the run times statistics increase as the problem size increases.

\begin{table}[H]
	\centering
	\caption{Run times of BP\&R (Seconds)}
	{
	\begin{tabular}{ccccc}
		\hline
		Size       & Mean    & Standard deviation & Max     & Min    \\ \hline
		Small      & 0.164   & 0.222              & 2.224   & 0.007  \\
		Medium     & 3.544   & 29.319             & 501.235 & 0.014  \\
		Large      & 60.540  & 146.494            & 509.818 & 0.109  \\
		Very Large & 253.994 & 188.356            & 532.686 & 15.300 \\ \hline
	\end{tabular}}
	\label{Table_6}
\end{table}

An analysis of instances unsolved by BP\&R is presented in Table \ref{Table_7}. This table shows how unsolved instance are divided in terms of problem characteristics. As seen in this table, for problem of medium size, all of the unsolved instances have a mid range order strength (0.6) and a central distribution of task time. Moreover, problems with bottlenecks (BN) in their precedence graph structure tend to be harder to solve for BP\&R. For problems of largers size, problems with central task time distribution still make up the majority of unsolved instances. In terms of order strength, problems with mid range to low order strength tend to be more challenging for BP\&R to solve. Moreover, in terms of graph structure, mixed problems appear to be harder to solve.

\begin{table}[H]
	\centering
	\caption{Unsolved problem statistics}
	{
	\begin{tabular}{ccccccccccccc}
		\hline
		\multirow{2}{*}{Size} &  & \multicolumn{3}{c}{Structure} &  & \multicolumn{3}{c}{Order Strength} &  & \multicolumn{3}{c}{Peak Location} \\ \cline{3-5} \cline{7-9} \cline{11-13} 
		&  & BN       & CH      & MIX      &  & 0.2        & 0.6       & 0.9       &  & Bottom    & Central   & Bimodal   \\ \hline
		Medium                &  & 11       & 2       & 1        &  & 0          & 14        & 0         &  & 0         & 14        & 0         \\
		Large                 &  & 9        & 4       & 34       &  & 0          & 24        & 23        &  & 0         & 45        & 2         \\
		Very Large            &  & 51       & 51      & 77       &  & 76         & 76        & 27        &  & 0         & 175       & 4         \\ \hline
	\end{tabular}}
\label{Table_7}
\end{table}

\section{Conclusions and future research}\label{sec4}
In this paper we developed a branch, price and remember algorithm for the U shaped assembly line balancing problem. Computational results show that the proposed algorithm verifies the optimal solution for 255 of the 269 benchmark problems available in the literature. Previous best known exact algorithm is able to solve 233 of the 269 benchmark problems. The proposed approach uses column generation to obtain tight bounds that are computed reasonably fast. It also uses memory to avoid revisiting already visited nodes in the search tree. 
The proposed algorithm is designed for the U shaped assembly line balancing problem, but the idea of branch, price and remember can be used to obtain tight bounds for other types generalized assembly line balancing problems such as two sided lines, balancing and scheduling problems with sequence dependent setup times, learning etc. Therefore, as a possible future research in this area, we suggest using the ideas in this paper to develop branch and price algorithms for other types of assembly line balancing problem.\\
\textbf{References}

\bibliographystyle{plain}
\bibliography{Refs}

\end{document}